%% file: IO-Parikh.tex
\documentclass[smallextended]{svjour3}

\usepackage{pslatex}
\usepackage{hyperref}
\usepackage{amssymb}
\usepackage{amsmath}
\usepackage{mathrsfs}
\usepackage{bbm}
\usepackage{enumerate}

\usepackage{wasysym, txfonts}
\usepackage{tikz}
\usepackage{todonotes}

\usepackage{indentfirst}

\newcommand{\imp}{\Rightarrow}
\newcommand{\iosub}[3]{\mbox{$#1[#2:=#3]_{IO}$}}
\newcommand{\vect}[1]{\overrightarrow{#1}}

\begin{document}

\title{On the effect of the IO-substitution on the Parikh image of semilinear AFLs\protect
}

\author{Pierre Bourreau
}
\institute{Institut f\"ur Sprache und Information\\
Heine-Heinrich Universit\"at D\"usseldorf\\
Universit\"atstr. 1\\
40225 D\"usseldorf, Germany\\
\email{pierre.bourreau@gmail.com
}}

\maketitle

\begin{abstract}
  Back in the 80's, the class of mildly context-sensitive formalisms was introduced so as to capture the syntax of natural languages. While the languages generated by such formalisms are constrained by the constant-growth property, the most well-known and used mildly context-sensitive formalisms like tree-adjoining grammars or multiple context-free grammars generate languages which verify the stronger property of being semilinear. In \cite{Bourreau-IOsubstitution}, the operation of IO-substitution was created so as to exhibit mildly-context sensitive classes of languages which are not semilinear although they verify the constant-growth property. In this article, we extend the notion of semilinearity, and characterise the Parikh image of the IO-MCFLs (\textit{i.e.}~languages which belong to the closure of MCFLs under IO-substitution) as universally-linear. Based on this result and on the work of Fischer on macro-grammars, we then show that IO-MCFLs are not closed under inverse homomorphism, which proves that the family of IO-MCFLs is not an abstract family of languages.
\end{abstract}


\input{1-introduction}

\section{Semilinearity, constant-growth and IO-substitution}\label{section:basics}
\input{2.1-constant-growth}
\input{2.2-IO-substitution}

\section{IO-MCFLs have factorized Parikh images}\label{section:parikh}
\input{3.1-parikh-image}
\input{3.2-factorized}

\section{Non-closure of IO-MCFLs under inverse homomorphism}\label{section:inverse-homomorphism}
\input{4.1-std-form}
\input{4.2-fully-effective}
\input{4.3-a-linearity}
\input{4.4-inverse-homomorphism}

\input{5-conclusion}

\begin{acknowledgements}
This work was funded by the Deutsche Forchungsgemeinschaft, under the project SFB 991 ``Die Struktur von Repr\"asentationen in Sprache, Kognition und Wissenschaft''.

I am thankful to Sylvain Salvati for the motivating discussions on this topic; to Laura Kallmeyer for her insights on the notions of universally-linear sets; and to Christian Wurm who helped me to improve the formal definitions with his feedbacks. The responsibility for any mistakes contained herein rests solely on me.
\end{acknowledgements}

\bibliographystyle{apalike}
\bibliography{Biblio}

\end{document}

%% file: 1-introduction.tex
\section{Introduction}

The mathematical description of natural languages syntax is a problem which has captured the attention of scientists for years. Since the initial work of Chomsky and Schutzenberger~\cite{chomsky-hierarchy} on formal languages, it is now commonly accepted that the class of context-free languages is too weak to entirely capture the structure of syntax. This was first proved in \cite{shieber85} and \cite{huybregts84}, through examples in Swiss-German, and later on confirmed in \cite{michaelis97} and discussed in~\cite{Kobele06Diss}. At the same time, \cite{Joshi85} defined a new class of formalisms which he called \emph{mildly context-sensitive}, in an attempt to answer the question: ``How much context-sensitivity is needed to provide reasonable structural descriptions?''; such formalisms are defined through the following conditions:
\begin{enumerate}
  \item the class of generated languages must encompass the class of context-free languages;
  \item they must take into account some limited cross-serial dependencies;
  \item they must be recognisable in polynomial-time;
  \item and, the generated languages must verify the constant-growth property;
\end{enumerate}

While this definition and the answer it gives to the initial question is under debate, we focus on the fourth point of the definition and the notion of constant-growth property. Indeed, many mildly context-sensitive formalisms are known to verify the stronger property of generating semilinear languages. It is for instance the case of tree-adjoining grammars, multiple context-free grammars~\cite{seki91} (or alternatively linear context-free rewriting systems~\cite{Vijay-ShankerWeirJoshi:87}), or derivational minimalist grammars\cite{stabler96,Michaelis98}. In the following work, we investigate the gap between semilinear languages and the ones which verify the constant-growth property.

In \cite{Bourreau-IOsubstitution}, an operation on languages called IO-substitution was defined. This operation allows one to enrich a class of languages with a limited copying mechanism. IO-substitution can indeed be seen as a bounded copying operation on strings. In this preliminary work, \cite{Bourreau-IOsubstitution} proved three main properties. First, given an abstract family of semilinear languages $\mathbf{L}$, its closure under IO-substitution $\mathbf{IO(L)}$ forms a family of languages which is closed under union, concatenation, homomorphism and intersection with regular sets; an open question is therefore to prove whether $\mathbf{IO(L)}$ is an abstract family of languages. Moreover, it was proved that if the languages in $\mathbf{L}$ verify the constant-growth property, so do the languages in $\mathbf{IO(L)}$. Finally, in the special case where $\mathbf{L}=\mathbf{MCFL}$, the class of multiple context-free languages, the authors showed that any language in $\mathbf{IO(MCFL)}$ can be recognised in polynomial-time; these first results lead to considering the formalisms which generated languages falls within $\mathbf{IO(MCFL)}$ as candidates for being mildly context-sensitive.

In the present article, we investigate a precise characterisation of the Parikh image of languages in $\mathbf{IO(L)}$, where $\mathbf{L}$ is a family of semilinear languages (\textit{i.e.}~as a particular case, the results we obtain apply when $\mathbf{L}$ is the family of regular, context-free, or multiple context-free languages of strings). In order to do so, we extend the notion of semilinearity in a natural way, by defining \emph{functional vector-sets}; from this definition, we consider two new characterisations for sets of vectors: \emph{existentially-semilinear sets} and \emph{universally-semilinear sets}, and show that the Parikh image of $\mathbf{IO(L)}$ falls within the second one, leading, as a corollary, to an alternative proof that such languages verify the constant-growth property. In the second part of the article, we give a proof of the non-closure of $\mathbf{IO(L)}$ under inverse homomorphism, where $\mathbf{L}$ is an abstract family of semilinear languages. This result, which is obtained thanks to the previous characterisation of the Parikh image for the considered languages and by reusing the main ideas of Fischer's proof of the non-closure of IO-macro grammars under inverse homomorphism, shows that $\mathbf{IO(L)}$ is not an abstract family of languages.

For simplicity and if not specified otherwise, any family of languages $\mathbf{L}$ will be considered as an abstract family of semilinear languages in the rest of the article.

The outline of this document is the following: section~\ref{section:basics} defines the fundamental notions needed from formal language theory: the Parikh image, semilinearity, the constant-growth property, and the IO-substitution. In section~\ref{section:parikh}, we introduce universal-semilinearity and existential-semilinearity as extended notions of semilinearity, and show that the Parikh image of languages in $\mathbf{IO(L)}$ falls into a class of sets for which the constant-growth property is verified. Finally, section~\ref{section:inverse-homomorphism} is dedicated to prove the non-closure of $\mathbf{IO(L)}$ under inverse homomorphism; this proof will also bring the opportunity to study new structural properties of $\mathbf{IO(L)}$.


%% file: 2.1-constant-growth.tex
\subsection{Formal languages, constant-growth and semilinearity}

We first introduce the notations for various usual notions related to formal languages. Given a set $\Sigma$ (called an alphabet), we write $\Sigma^\ast$ for the set of
words built on $\Sigma$, and $\epsilon$ for the empty word.  Given $w$
in $\Sigma^\ast$, we write $|w|$ for its length, and $|w|_a$ for the
number of occurrences of a letter $a$ of $\Sigma$ in $w$.  A language
on $\Sigma$ is a subset of $\Sigma^\ast$. Given a language $L$, we will speak of the alphabet $\Sigma$ of $L$ to designate any set such that $L \subseteq \Sigma^\ast$. Given two languages
$L_1,L_2\subseteq\Sigma^\ast$, $L_1\cdot L_2$, the concatenation of
$L_1$ and $L_2$, is the language $\{w_1w_2\mid w_1\in L_1\land w_2\in
L_2\}$; the union of $L_1$ and $L_2$ is written $L_1+L_2$.

We write $\mathbb{N}$ for the set of natural numbers. For a finite
alphabet $\Sigma$, $\mathbb{N}^\Sigma$ is the set of vectors whose
coordinates are indexed by the letters of $\Sigma$. Vectors will be noted $\vect{v}$, and a $n$-dimensional vector ($n \in \mathbb{N}$) will be written $\langle c_1, \dots, c_n \rangle$ (where $c_1, \dots, c_n \in \mathbb{N}$) when we wish to exhibit the values of the vector on each of its dimension. Given $a \in
\Sigma$ and $\vect{v}\in \mathbb{N}^{\Sigma}$, $\vect{v}[a]$ will denote
the value of $\vect{v}$ on the dimension $a$.

In \cite{Joshi85}, the constant-growth property was introduced as a condition languages generated by mildly context-sensitive formalisms must verify. This condition expresses some constraints on the distribution of the length of the words in a language:

\begin{definition}[Constant-growth]\label{def:constant_growth}
  A language $L \subseteq \Sigma^*$ is said to be \emph{constant-growth} if there exist $k,c\in\mathbb{N}$ such that, for every $w \in L$, if $|w| > k$, then there is $w' \in L$ for which $|w|<|w'|\leq |w|+c$.

\end{definition}

As mentioned in the introduction, most of the mildly context-sensitive generative formalisms commonly
used in modeling natural language syntax generate languages which verify the stronger property of semilinearity, which is based on the following notion of the Parikh image.

\begin{definition}[Parikh image]
  Let us consider a word $w$ in a language $L \subseteq \Sigma^*$. The
  \emph{Parikh image of $w$}, written $\vect{p}(w)$ is the vector of $\mathbb{N}^\Sigma$ such that, for every $a \in \Sigma$,
  $\vect{p}(w)[a] = |w|_a$. The Parikh image of $L$ is defined as
  $\vect{p}(L) = \{\vect{p}(w)\ |\ w \in L\}$.
\end{definition}

\begin{definition}[Semilinearity]
  A set $V$ of vectors of $\mathbb{N}^\Sigma$ is said \emph{linear}
  when there are vectors $\vect{v_0}$, \ldots, $\vect{v_n}$ in
  $\mathbb{N}^\Sigma$ such that $V=\{\vect{v_0}+k_1\vect{v_1}+ \ldots
  + k_n\vect{v_n} \mid k_1, \ldots, k_n \in \mathbb{N}\}$.

  A set of vectors is said \emph{semilinear} when it is a finite union of linear sets. 
\end{definition}

Given two sets of vectors $V_1$ and $V_2$ of $\mathbb{N}^k$, for $k \in \mathbb{N}$, we will note $V_1+V_2$
the set $\{\vect{v_1}+\vect{v_2} \mid \vect{v_1} \in V_1, \vect{v_2}
\in V_2\}$. Similarly, given $c \in \mathbb{N}$ and a set of vectors
$V$ of $\mathbb{N}^k$, we will write $cV=\{c \vect{v} \mid \vect{v} \in V\}$.

\begin{definition}
  A language $L$ is said \emph{semilinear} when $\vect{p}(L)$ is a
  semilinear set.
\end{definition}

Well-known  classes of semilinear languages are the class \textbf{RL} of regular languages, the class \textbf{CFL} of context-free languages, the class \textbf{yTAL} of yields of tree-adjoining languages or the class \textbf{MCFL} of multiple context-free languages.

\begin{definition}
Given a class of languages $\mathbf{L}$ and a class of sets of vectors $\mathcal{V}$, we say that \textbf{L} is full for $\mathcal{V}$ if for every language $L \in \mathbf{L}$, $\vect{p}(L) \in \mathcal{V}$; $\mathbf{L}$ is said complete for $\mathcal{V}$ if for every $V \in \mathcal{V}$, there exists $L \in \mathbf{L}$ such that $\vect{p}(L)=V$.
\end{definition}

It is known that $\mathbf{RL}$ is full and complete for the class of semilinear sets. Consequently, \textbf{CFL}, \textbf{yTAL} and \textbf{MCFL} also verify this property as they are full for the class of semilinear sets and include all languages in \textbf{RL}.

Given two alphabets $\Sigma_1$ and $\Sigma_2$, a string homomorphism
$h$ from $\Sigma_1^\ast$ to $\Sigma_2^\ast$ is a function such that
$h(\epsilon) = \epsilon$ and $h(w_1w_2) = h(w_1)h(w_2)$, where $w_1, w_2 \in \Sigma_1^\ast$.  Given
$L\subseteq \Sigma_1^*$, we write $h(L)$ for the language $\{h(w) \in \Sigma_2^*\mid
w\in L\}$.

  Let us consider a class $\mathbf{L}$ of languages. Given an $n$-ary operation
  $\mathtt{op}:(\Sigma^*)^n \to \Sigma^*$ on strings (where $n \in
  \mathbb{N}$), we say that $\mathbf{L}$ is closed under $\texttt{op}$ if
  for every $L_1, \ldots, L_n \in \mathbf{L}$, $\texttt{op}(L_1,
  \ldots, L_n) \in \mathbf{L}$.

\begin{definition}[AFLs]
  A class of languages \textbf{L} is called an \emph{abstract family of languages} (written AFL for concision) if it is closed under union, concatenation, Kleene star, (alphabetic) homomorphism, inverse (alphabetic) homomorphism and intersection with regular sets.
\end{definition}

The previously defined classes \textbf{RL}, \textbf{CFL}, \textbf{yTAL}, and \textbf{MCFL} are known to be AFLs.


%% file: 2.2-IO-substitution.tex
\subsection{IO-substitution: going beyond semilinearity}

In \cite{Bourreau-IOsubstitution}, the operation of IO-substitution was defined so as to enrich languages with a limited copying operation.

\begin{definition}[IO-substitution]
  Let us consider the alphabets $\Sigma_1$ and $\Sigma_2$, and two languages
  $L_1 \subseteq \Sigma_1^*$ and $L_2 \subseteq \Sigma_2^*$. Given a word $w \in L_2$, and a symbol $a \in \Sigma_1$, we define the homomorphism $io_{a, w}$ based on the function
  $$\begin{array}{cccl}
    io_{a, w} : & \Sigma_1 & \to & \Sigma_2^*\\
    & c & \mapsto & \begin{cases} w & \text{if } c=a\\ c & \text{otherwise}
  \end{cases}
\end{array}$$

 We define the relation of IO-substitution as
 \begin{center}
   $L_1[a:=L_2]_{IO} \rightarrow L$ iff $L=\bigcup_{w \in L_2}io_{a, w}(L_1)$
 \end{center}
 and we call $L_1[a:=L_2]_{IO}$ the \emph{IO-substitution of $a$ by $L_2$ in $L_1$}.
\end{definition}

Note, if for every word $w \in L_1$, $a$ has no occurrence in $w$ (\textit{i.e.}~$|w|_a=0$) then $L_1[a:=L_2]_{IO} \rightarrow L_1$ is verified.

In the rest of the document, we will use the notation $d \rightarrow^\ast L$ iff, either $d = L$; or $d=d_1[x:=d_2]_{IO}$, such that $d_1 \rightarrow^* L_1$, $d_2 \rightarrow^\ast L_2$ and $L_1[x:=L_2]_{IO} \rightarrow L$. 

\begin{example}\label{ex:io-languages}
Let us consider the languages $L_1=a^*$ and $L_2=ab+c$; the language $L$ such that $L_1[a:=L_2]_{IO} \rightarrow L$ is then defined as $(ab+c)^*$. In this case, $L$ is a regular language, just like $L_1$ and $L_2$.

Another more interesting example is $L=\{a^p \mid p \text{ is not a prime number}\}$, which verifies $xx^\ast x[x:= aa^\ast a]_{IO} \rightarrow L$. Such a language is not semilinear since its Parikh image is equal to $\{nm \langle 1 \rangle \mid n,m >1\}$. Therefore $L$ does not belong to \textbf{RL}, while $L_1$ and $L_2$ do.
\end{example}

\begin{definition}[\textbf{IO(L)}]
  Given a class of languages $\mathbf{L}$, we define the class $IO_n(\mathbf{L})$ by induction on $n \in \mathbb{N}$ as

\begin{enumerate}
\item $IO_0(\mathbf{L}) =  \mathbf{L}$
\item for $n \geq 0$, $$IO_{n+1}(\mathbf{L})=IO_n(\mathbf{L}) \cup \bigcup_{L_1, L_2 \in IO_n(\mathbf{L})} \bigcup_{x \in \Sigma_1}\{L \mid L_1[x:=L_2]_{IO} \rightarrow L\}$$
where $\Sigma_1$ is the alphabet of $L_1$
\end{enumerate}

The smallest class of languages containing $\mathbf{L}$ and closed under IO-substitution is defined by $IO(\mathbf{L}) = \{L \in IO_n(\mathbf{L}) \mid n \in \mathbb{N}\}$
\end{definition}

We introduce the notion of derivations and derivation trees associated to a language in \textbf{IO(L)}.

\begin{definition}[Derivation]
  Given a language $L$ in $\mathbf{IO(L)}$, we define the set of \emph{derivations} $\mathcal{D}_L$ and the set of \emph{derivation trees} $\mathcal{T}_L$ associated to $L$ as the smallest sets such that:
  \begin{itemize}
    \item if $L \in \mathbf{L}$, then $L \in \mathcal{D}_L$; and $t=n \in \mathcal{T}_L$, where $n$ is a node labelled with $L$;
    \item if $L_1[x:=L_2]_{IO} \rightarrow L$ for $L_1, L_2 \in \mathbf{IO(L)}$, then, $\{d_1[x:=d_2]_{IO} \mid d_1 \in \mathcal{D}_{L_1}, d_2 \in \mathcal{D}_{L_2}\} \subseteq \mathcal{D}_L$; and given $n$ a node labelled with $x$, $\{n(t_1, t_2) \mid t_1 \in \mathcal{T}_{L_1}, t_2 \in \mathcal{T}_{L_2}\} \subseteq \mathcal{T}_L$.
  \end{itemize}
\end{definition}


\begin{example}
  Let us consider some languages $L_i \in \mathbf{L}$ for $1 \leq i \leq 5$ and $$((L_1[x_1:=L_2]_{IO})[x:=(L_3[x_2:=L_4]_{IO})]_{IO})[y:=L_5]_{IO} \rightarrow^* L$$
\begin{figure}
 \begin{center}
    \begin{tikzpicture}[level distance=7mm,
      level 1/.style={sibling distance=3cm},
      level 2/.style={sibling distance=3cm},
      level 3/.style={sibling distance=1.7cm},
      level 4/.style={sibling distance=1cm}]
      \tikzstyle{edge from parent}=[draw]
      \node (b0) {$y$}
      child {node (b1) {$x$}
        child {node (b11) {$x_1$}
          child {node (b111) {$L_1$}}
          child {node (b112) {$L_2$}}}
        child {node (b12) {$x_2$}
          child {node (b121) {$L_3$}}
          child {node (b122) {$L_4$}}}}
      child {node (b2) {$L_5$}}
      ;
    \end{tikzpicture}
  \end{center}
  \caption{Example of a derivation tree associated to a language in \textbf{IO(L)}}
  \label{fig:tree-representation-io}
\end{figure}
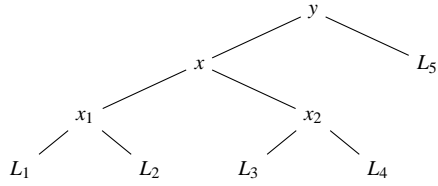
 The corresponding derivation tree is represented by the binary tree in Figure~\ref{fig:tree-representation-io}.
\end{example}

From the definition above, it is obvious that the set of derivations and the set of derivation trees are isomorphic. We will not use the notion of derivation tree in the rest of the article, but we will allow ourself to speak about a subderivation to speak about the derivation associated to a subtree of a given derivation tree.

In the rest of the document, we adopt a convention of left-associativity for the IO-substitution operation: a derivation $L_1[x_1:=L_2]_{IO}[x_2:=L_3]_{IO}$ will denote the derivation $(L_1[x_1:=L_2]_{IO})[x_2:=L_3]_{IO}$.

As pointed out in \cite{Bourreau-IOsubstitution}, the IO-substitution operation can be seen as a restriction of the copying power of IO-macro grammars in \cite{fischerphd,fischer1968grammars}. Indeed, the authors gave a grammatical formalism in terms of abstract categorial grammars \cite{degroote01,muskens01} which generates languages in \textbf{IO(MCLF)}, and the construction exhibits the use of copies in a non-recursive way, \textit{i.e.} the use of a bounded number of copies; this restriction leads, for instance, to exclude languages like $L_{sq}\{a^{n^2} \mid n \in \mathbb{N}\}$ from \textbf{IO(MCFL)}, while such a language is known to be generated by IO-macro grammars (and also by parallel MCFGs \cite{seki91}, another formalism which enriches MCFGs with deletion and copying operations). One will note that $L_{sq}$ is not a constant-growth language.

The property is even stronger as the IO-substitution preserves the constant-growth property of languages under some constraints, as given in the following theorem:

\begin{theorem}\cite{Bourreau-IOsubstitution}
  Given \textbf{L} an abstract family of semilinear languages:
  \begin{itemize}
    \item \textbf{IO(L)} is a family of  constant-growth languages
    \item \textbf{IO(L)} is closed by homomorphism, intersection with regular sets, finite union and concatenation.
  \end{itemize}
\end{theorem}

We now investigate a precise characterisation of the Parikh images of languages in \textbf{IO(L)}, and give an alternative proof of the constant-growth property for the languages in this class. In order to do so, we will give a natural extension of semilinear sets, in terms of functions.


%% file: 3.1-parikh-image.tex
\subsection{Constant-growth and parametrised-growth}

As mentioned in the previous section, the Parikh image of \textbf{IO(L)} goes beyond semilinear sets, while being captured by the notion of constant-growth. In the next section, we generalise the notion of semilinear sets to parametrised sets. Indeed, one can see a linear set $V = \{\vect{v_0}+x_1\vect{v_1}+ \dots + x_n\vect{v_n} \mid x_1, \dots, x_n\}$, where $\vect{v_0}, \dots, \vect{v_n} \in \mathbb{N}^p$, as the image of the function $f:\mathbb{N}^n \mapsto \mathbb{N}^p$ such that $f(x_1, \dots, x_n) = \vect{v_0}+x_1\vect{v_1}+ \dots + x_n\vect{v_n}$, hence parametrised by the variables $x_1, x_2, \dots, x_n$.

\begin{definition}[Vector function]
  We define a (multi-valued) \emph{vector function} $F$ as a function whose domain is $\mathbb{N}^n$ and whose codomain is $\mathbb{N}^p$, and such that there exists $m \in \mathbb{N}$ for which:
$$F(x_1, \dots, x_n)=\sum_{1 \leq j \leq m }f_j(x_1, \dots, x_n)\vect{v_j}$$
where $f_j: \mathbb{N}^n \to \mathbb{N}$ and $\vect{v_j}$ is a vector of $\mathbb{N}^p$, for all $1 \leq j \leq m$.
\end{definition}

A linear function $F:\mathbb{N}^n \to \mathbb{N}^p$ is therefore a particular case of vector functions, because it has the form $F(x_1, \dots, x_n)=\vect{v_0} + \sum_{1 \leq j \leq n}x_j\vect{v_j}$.

\begin{definition}[Functional vector-set]
  A \emph{functional vector-set} $E$ is defined as the finite union of the image of some vector functions: $E=\bigcup_{1 \leq i \leq k} \mathrm{Im}(F_i)$.

  A language which Parikh image is a functional vector-set will be called a \emph{language with parametrised growth}.
\end{definition}

Following this definition, a semilinear set can be alternatively defined as a finite union $\bigcup_{1 \leq i \leq k} \mathrm{Im}(F_i)$, where $F_i$ is a linear function, for every $1 \leq i \leq k$. In the rest of the article, and without loss of generality, when considering a functional vector-set defined by $\bigcup_{1 \leq i \leq k} \mathrm{Im}(F_i)$, we will assume that the functions $F_i$, $1 \leq i \leq k$ share the same codomain.

The definition above is too general to only capture the Parikh image of languages which verify the constant-growth property, as shown with the following example.

\begin{example}\label{ex:parameter-growth}
The language $L_{sq}=\{a^{n^2} \mid n \in \mathbb{N}\}$ has a parametrised growth: its Parikh image is given by $\mathrm{Im}(F)$, where $F(x)=x^2 \langle 1 \rangle$;it is easy to see that $L_{sq}$ does not verify the constant-growth property.
\end{example}

We next define specific functional vector-sets which approximates the ideas behing the constant-growth property.

\begin{definition}[$i$-linear vector function]
  Given a vector function $F:\mathbb{N}^n \to \mathbb{N}^m$ and $1 \leq i \leq n$, $F$ is said \emph{$i$-linear}, if given $F(x_1, \dots, x_n) = \sum_{1 \leq j \leq m} f_j(x_1, \dots, x_n)\vect{v_j}$,  for every $1 \leq j \leq m$ we have:
  $$f_j(x_1, \dots, x_n)=A_j(x_1, \dots, x_{i-1}, x_{i+1}, \dots, x_n)x_i + B_j(x_1, \dots, x_{i-1}, x_{i+1}, \dots, x_n)$$
where $A_j, B_j:\mathbb{N}^{n-1} \to \mathbb{N}$.

We say that $F$ is:
\begin{itemize}
  \item \emph{existentially-linear} if there exists $1 \leq i \leq n$, for which $F$ is $i$-linear.
  \item \emph{universally-linear} if for every $1 \leq i \leq n$, $F$ is $i$-linear.
  \end{itemize}
\end{definition}

These definitions are then naturally extended to sets of vectors:

\begin{definition}[Existentially- and Universally-semilinear sets]
 
 A functional vector-set $E =  \bigcup_{1 \leq i \leq k} \mathrm{Im}(F_i)$ is said to be:
  \begin{itemize}
    \item existentially-semilinear if there exists $1 \leq i \leq k$ for which $F_i$ is existentially-linear. 
    \item universally-semilinear if for every $1 \leq i \leq k$, $F_i$ is universally-linear.
 \end{itemize}
\end{definition}

It is then obvious that a universally-linear vector function is existentially-linear. Similarly, a universally-linear set is existentially-linear.

\begin{lemma}
  Given a language $L$, if $L$ is existentially-linear then $L$ is constant-growth.
\end{lemma}
\begin{proof}
  Let us consider such a language $L$; then $\vect{p}(L)=\bigcup_{1 \leq i \leq k} \mathrm{Im}(F_i)$ and the existence of some  $1 \leq i \leq k$ such that $F_i$ is existentially-linear, \textit{i.e.}, for $\mathrm{Dom}(F_i)=\mathbb{N}^n$ and $F(x_1, \dots, \dots, x_n) = \sum_{1 \leq j \leq m} f_j(x_1, \dots, x_n) \vect{v_j}$ there exists $1 \leq l \leq n$ such that, for every $1 \leq j \leq m$:
$$f_j(x_1, \dots, x_n)=A_j(x_1, \dots, x_{l-1}, x_{l+1}, \dots, x_n)x_l + B_j(x_1, \dots, x_{l-1}, x_{l+1}, \dots, x_n)$$
Let us consider $c_1, \dots, c_n \in \mathbb{N}$. We then write
$$A=A_j(c_1, \dots, c_{l-1}, c_{l+1}, \dots, c_n)$$
$$B=B_j(c_1, \dots, c_{l-1}, c_{l+1}, \dots, c_n)$$
$$K=Ac_l+B$$

Then we can build a (increasing) sequence of words $(w_i)_{i \in \mathbb{N}}$, such that $|w_i|=K+iA$, for every $i \in \mathbb{N}$. Therefore, given a word $w \in L$ such that $|w|>K$, we can find $i \in \mathbb{N}$ such that $|w_i| < |w| \leq |w_{i+1}|$, \textit{i.e.} $|w_i| < |w| \leq |w_{i}|+C$.
\end{proof}

The family of existentially-linear sets seems to be the biggest family of constant-growth sets of vectors definable from the definition of functional vector-sets. Interestingly enough, existentially-linear sets show how big is the gap between constant-growth and semilinear languages. As an example, consider the language $\{a^{n^2}b^mc^{nm} \mid n,m \in \mathbb{N}\}$, whose Parikh image is given by $\mathrm{Im}(F)$, where $F(x_1, x_2)=x_1^2 \langle 1, 0, 0\rangle +x_2 \langle 0, 1, 0 \rangle + x_1x_2 \langle 0, 0, 1 \rangle$. Then $F$ is existentially-linear (for $x_2$) and therefore constant-growth.

Languages which Parikh images are existentially-linear vector sets can be seen as languages which have a ``linear sub-basis''. Indeed, the definition of an existentially-linear set of vectors states that an infinite subset of it verifies a linear growth.  As a particular case and if we only consider the formal definition of mildly context-sensitivity, formalisms which allow copy mechanisms should be considered as candidates for mildly-context sensitive formalisms as soon as they ensure such a linear sub-basis in the languages generated. Such a property might be interesting in the description of natural language syntax, in case one wants to describe ellipsis through copying operations~\cite{Sarkar-Joshi1996,Kobele07ParsingEllipsis,bourreau13}, or to integrate copying phenomena appearing, for instance in Yes-No questions in Mandarin~\cite{radzinski90}, or in relatives in Bambara~\cite{culy87}.

Finally, we can remark that the notion of universally-linear language seems to be closer to the ideas expressed in the following revision of the constant-growth property of~\cite{kallmeyer10}, where a language is constant-growth if there exists a constant $c \in \mathbb{N}$, such that for every word $w \in L$ verifying $|w|>c$, there are vectors $\vect{v_1}$ and $\vect{v_2}$ for which $\vect{p}(w)=\vect{v_1}+\vec{v_2}$ and for every $k \geq 1$, $\vect{v_1} + k\vect{v_2} \in \vect{p}(L)$.
Indeed, any word $w$ in a universally-linear language belongs to a sublanguage which verifies the constant-growth property, and this sublanguage is given by the vector function associated to $w$.


%% file: 3.2-factorized.tex
\subsection{Factored Parikh image}\label{part:factorized-parikh}

We now give a precise characterisation of the Parikh images of languages in \textbf{IO(L)}. We will prove that such images are particular cases of universally-linear sets. This result leads to a proof of the constant-growth property of these languages, which differs from the one given in~\cite{Bourreau-IOsubstitution}.

In what follows, we denote by $\mathcal{F}(\mathbb{N}^n, \mathbb{N}^m)$ the set of vector functions whose domain is $\mathbb{N}^n$ and whose codomain is $\mathbb{N}^m$. Moreover, given a vector $\vect{v}=\langle v_1, \dots, v_n \rangle$ on $\mathbb{N}^n$ and an integer $1 \leq k \leq n$, we write $\vect{v}[k] = v_k$ and $\vect{v}|_k = \langle v_1, \dots, v_{k-1}, 0, v_{k+1}, \dots, v_n \rangle$.

\begin{definition}
  Let us consider a vector function $F:\mathbb{N}^n \mapsto \mathbb{N}^m$, and $1 \leq k \leq m$, we define the functions:
  \begin{itemize}
    \item $\mathrm{Wtt}_k: \mathcal{F}(\mathbb{N}^n, \mathbb{N}^m) \to \mathcal{F}(\mathbb{N}^{n}, \mathbb{N}^{m})$ such that $\mathrm{Wtt}_k(F)(x_1, \dots, x_n) = F(x_1, \dots, x_n)|_k$
    \item $\mathrm{On}_k: \mathcal{F}(\mathbb{N}^n, \mathbb{N}^m) \to \mathcal{F}(\mathbb{N}^{n}, \mathbb{N})$ such that $\mathrm{On}_k(F)(x_1, \dots, x_n)= F(x_1, \dots, x_n)[k]$.
  \end{itemize}
\end{definition}

When a vector will be associated to the Parikh image of a language, we will allow ourselves to index these two functions by the letter corresponding to the dimension, and use the notations $\mathrm{Wtt}_x$ and $\mathrm{On}_x$.

From the functions $\mathrm{Wtt}_k$ and $\mathrm{On}_k$, we define the following notion of a factored-semilinear Parikh image.

\begin{definition}[Factored vector function]
A vector function $F:\mathbb{N}^n \to \mathbb{N}^m$ is said to be \emph{factored} if the following induction stands:
\begin{enumerate}
  \item $F$ is a linear function, or
  \item there exist $1 \leq k \leq m$, $F_1:\mathbb{N}^{n_1} \to \mathbb{N}^m$ and $F_2:\mathbb{N}^{n_2} \to \mathbb{N}^m$, factored vector functions such that $n=n_1 + n_2$ and $$F(x_1, \dots, x_{n_1+n_2})=\mathrm{Wtt}_k(F_1)(x_1, \dots, x_{n_1}) + \mathrm{On}_k(F_1)(x_1, \dots, x_{n_1})F_2(x_{n_1+1}, \dots x_{n_1+n_2})$$
\end{enumerate}
\end{definition}

In the rest of the document, we allow ourselves to write $\mathrm{Wtt}_k(F_{j_1}) + \mathrm{On}_k(F_{j_1})F_{j_2}$ for a function as in 2. in the definition above.

\begin{definition}[Factored-semilinear set]
A vector set $E$ is \emph{factored-semilinear}, if the following induction stands
\begin{enumerate}
  \item it is a semilinear set, or
  \item it is of the form $$\bigcup_{1 \leq j_1 \leq m_1} \bigcup_{1 \leq j_2 \leq m_2} \mathrm{Im}(\mathrm{Wtt}_k(F_{j_1}) + \mathrm{On}_k(F_{j_1})F_{j_2})$$ where $\bigcup_{1 \leq j_1 \leq m_1}F_{j_1}$ and $\bigcup_{1 \leq j_2 \leq m_2}F_{j_2}$ are factored-semilinear sets; and for $\mathbb{N}^n$ the codomain of $F_{j_1}$ (for every $1 \leq j_1 \leq m_1$), we have $1 \leq k \leq m$.
\end{enumerate}
\end{definition}

We now prove that $\mathbf{IO(L)}$ is full and complete for factored-semilinear sets, if $\mathbf{L}$ is full and complete for linear sets.

\begin{proposition}\label{th:IO-factorized}
  Every language $L \in \mathbf{IO(L)}$ has a factored-semilinear Parikh image.
\end{proposition}
\begin{proof}
  By definition, there exists $n \in \mathbb{N}$ such that $L \in IO_n(\mathbf{L})$. We proceed by induction on $n$:
  \begin{itemize}
    \item if $n=0$, then $L$ belongs to $\mathbf{L}$. By definition $\vect{p}(L)$ is a semilinear set, therefore a factored-semilinear vector set.
    \item otherwise, there exist $L_1, L_2 \in IO_{n-1}(\mathbf{L})$ such that $L_1[x:=L_2]_{IO} \rightarrow L$. Then, for every word $w \in L$, there exist $w_1 \in L_1$ and $w_2 \in L_2$ such that $w=io_{x, w_2}(w_1)$. 
Let us consider $\bigcup_{1 \leq k_1 \leq m_1}\mathrm{Im}(F_{k_1})$ and $\bigcup_{1 \leq k_2 \leq m_2}\mathrm{Im}(F_{k_2})$ as the Parikh images of $L_1$ and $L_2$ respectively. There exists $1 \leq i_1 \leq m_1$ and $1 \leq i_2 \leq m_2$, and $c_1, \dots, c_{n_1}, c'_1, \dots, c'_{n_2}$ in $\mathbb{N}$ such that $F_1(c_1, \dots, c_{n_1})$ is the Parikh image of $w_1$ and $F_2(c'_1, \dots, c'_{n_2})$, the Parikh image of $w_2$. Then it is easy to see that the Parikh image of $w$ is  $$\mathrm{Wtt}_x(F_1)(c_1, \dots, c_{n_1}) + \mathrm{On}_x(F_1)(c_1, \dots, c_{n_1})F_2(c'_1, \dots, c'_{n_2})$$
Moreover, $L_1$ and $L_2$ belong to $IO_{n-1}(\mathbf{L})$ and, by induction hypothesis, their respective Parikh image $\bigcup_{1 \leq k_1 \leq m_1}\mathrm{Im}(F_{k_1})$ and $\bigcup_{1 \leq k_2 \leq m_2}\mathrm{Im}(F_{k_2})$ are factored-semilinear sets. Therefore, $L$ has a factored-semilinear Parikh image.
  \end{itemize}
\end{proof}

\begin{proposition}\label{prop:factorized->io}
If a class of languages $\mathbf{L}$ is complete for semilinear sets, then $\mathbf{IO(L)}$ is complete for factored-semilinear sets.
\end{proposition}
\begin{proof}
Let us consider $E=\bigcup_{1 \leq i \leq n}\mathrm{Im}(H_i)$ a factored-semilinear sets and proceed by induction on it:
  \begin{itemize}
    \item if for every $1 \leq i \leq n$, $H_i$ is a linear function, then $E$ is a semilinear set; hence, by hypothesis, there is a language $L \in \mathbf{L}$ such that $\vect{p}(L)=E$.
    \item otherwise, we have $E=\bigcup_{1 \leq k_1 \leq m_1} \bigcup_{1 \leq k_2 \leq m_2} \mathrm{Im}(\mathrm{Wtt}_x(F_{k_1}) + \mathrm{On}_x(F_{k_1})F_{k_2})$, where $E_1 = \bigcup_{1 \leq k_1 \leq m_1}\mathrm{Im}(F_{k_1})$ and $E_2 = \bigcup_{1 \leq k_2 \leq m_2}\mathrm{Im}(F_{k_2})$ are factored-semilinear sets. By induction hypothesis, there exist $L_1$ and $L_2$ in $\mathbf{IO(L)}$ such that $\vect{p}(L_1)=E_1$ and $\vect{p}(L_2)=E_2$. It is then easy to see that, given $L$ such that $L_1[x:=L_2]_{IO} \rightarrow L$, we have $\vect{p}(L)=E$.
  \end{itemize}
\end{proof}

From Propositions\ref{th:IO-factorized} and \ref{prop:factorized->io}, we can deduce the next corollary, which establishes the strong relation between $\mathbf{IO(RL)}$ (or $\mathbf{IO(CFL)}$, $\mathbf{IO(MCFL)}$) and factored-semilinear languages.

\begin{corollary}
A vector set $E$ is a factored-semilinear set iff there exists $L \in \mathbf{IO(L)}$ such that $\vect{p}(L)=E$, where $\mathbf{L}$ is full and complete for semilinear sets.
\end{corollary}

This corollary leads to an alternative proof of the constant-growth property for languages in $\mathbf{IO(L)}$; it suffices to show that factored-semilinear set are existentially-linear; we prove the stronger statement that these languages are universally-linear.

\begin{theorem}
  Every factored-semilinear set is universally-linear.
\end{theorem}
\begin{proof}
  Let us consider an arbitrary factored-semilinear set $E$, and proceed by induction on it:
  \begin{itemize}
    \item if $E$ is a semilinear set $\bigcup_{1 \leq k \leq n} \mathrm{Im}(F_k)$ where $n \in \mathbb{N}$, then, because $F_k$ is a linear function for every $1 \leq k \leq n$, we directly obtain that $F_k$ is universally-semilinear.
    \item otherwise, we have $E=\bigcup_{1 \leq k \leq n}\bigcup_{1 \leq l \leq m} \mathrm{Im}(\mathrm{Wtt}_x(F_k)) + \mathrm{On}_x(F_k)G_l)$, where the sets $\bigcup_{1 \leq k \leq n}\mathrm{Im}(F_k)$ and $\bigcup_{1 \leq l \leq m}\mathrm{Im}(G_l)$ are factored-semilinear sets.

Let us consider a function $H_{ij}=\mathrm{Wtt}_x(F_i) + \mathrm{On}_x(F_i)G_j$, for some $1 \leq i \leq n$ and some $1 \leq j \leq m$,  and functions $F_i:\mathbb{N}^{p_i} \to \mathbb{N}^p$, $G_j:\mathbb{N}^{p_j} \to \mathbb{N}^p$ and $H_{ij}:\mathbb{N}^{p_i+p_j} \to \mathbb{N}^p$. We consider $H_{ijq}(x)=H_{ij}(a_1, \dots, a_{q-1}, x, a_{q+1}, \dots, a_{p_1+p_2})$, where $a_r \in \mathbb{N}$ for every $r$ such that $1 \leq r \leq q-1$ or $q+1 \leq r \leq p_1+p_2$. We show that $H_{ij}$ is $q$-linear (\textit{i.e.}~$H_{ijq}$ is linear)
        \begin{itemize}
          \item if $1 \leq q \leq p_1$, then:
            \begin{align*}
              H_{ijq}(x) = & \mathrm{Wtt}_x(F_i)(a_1, \dots, a_{q-1}, x, a_{q+1}, \dots, a_{p_1}) + & \\
              & \mathrm{On}_x(F_i)(a_1, \dots, a_{q-1}, x, a_{q+1}, \dots, a_{p_1})G_j(a_{p_1+1}, \dots, a_{p_1+p_2}) &\\
               = & (Ax+B)+(A'x+B')C \text{\hspace{55pt} $F_i$ is universally-linear} &\\
               = & (A+A'C)x+(B+B'C) &
            \end{align*}
            and $H_{ij}$ is therefore $q$-linear.
            \item if $p_1+1 \leq p \leq p_1+p_2$, then:
            \begin{align*}
              H_{ijq}(x) = & \mathrm{Wtt}_x(F_i)(a_1, \dots, a_{p_1}) + &\\
              & \mathrm{On}_x(F_i)(a_1, \dots, a_{p_1})G_j(a_{p_1+1}, \dots, a_{q-1}, x, a_{q+1}, dots, a_{p_1+p_2}) &\\
              = & A+A'(Bx+C) \text{\hspace{55pt}as $G_j$ is universally-linear} &\\
              = & A'Bx+(A+A'C) &
            \end{align*}
            and again $H_{ij}$ is $q$-linear.
          \end{itemize}
          Therefore, $H_{ijq}$ is linear for every $p_1 \leq k \leq p_2$, hence $H_{ij}$ is universally-linear. We conclude that $E$ is a universally-linear set.
      \end{itemize}
\end{proof}

We therefore proved that, given \textbf{L} an abstract family of languages, full and complete for semilinear sets, \textbf{IO(L)} is full and complete for factored-semilinear sets. It is then easy to see that \textbf{IO(L)} is not complete for universally-linear sets. Indeed, consider the set $\mathrm{Im}(F)$ where $F(x_1, x_2, x_3)=x_1x_2\langle 1, 0, 0 \rangle + x_2x_3\langle 0, 1, 0 \rangle + x_1x_3\langle 0, 0, 1 \rangle$. According to the definition, $F$ is universally-linear but not factored-semilinear. It is therefore an open question to define a formalism which is full and complete for universally-linear sets, or for existentially-linear sets.
We hope these two newly introduced classes of sets can be relevant in the study of other classes of languages.

In the next section, we show that \textbf{IO(MCFL)} is not an abstract family of languages, by proving it is not closed under inverse homomorphism. The proof is done similarly to the proof that IO macro-grammars are not closed under inverse homomorphism in \cite{fischerphd}, but differs in not being strongly connected to the formalism under study; instead, we will pay special attention on the effect of the IO-substitution on the properties of the Parikh image of languages in \textbf{IO(MCFL)}. As the proof is not strongly related to any grammatical formalism which generates \textbf{IO(MCFL)}, it will be directly extended to the non-closure under inverse homomorphism of \textbf{IO(RL)}, \textbf{IO(CFL)} or \textbf{IO(yTAL)}.


%% file: 4.1-std-form.tex
In \cite{Bourreau-IOsubstitution}, the closure of $\mathbf{IO(MCFL)}$ under homomorphism, concatenation, union and intersection with regular sets was proved. We here prove that the closure under inverse homomorphism is not satisfied, leading, as a corollary, to the proof that $\mathbf{IO(MCFL)}$ is not an abstract family of languages. In order to simplify the proof, we will first give some structural properties of $\mathbf{IO(MCFL)}$.

\subsection{Standard derivations for \textbf{IO(L)}}

In this section, we introduce a first specific form of derivations of a language in $\mathbf{IO(L)}$, and prove that every language in $\mathbf{IO(L)}$ can be derived thanks to such a derivation. The idea is to remove for a given derivation, every IO-substitution which is irrelevant (\textit{i.e.}~such that $L_1[x:=L_2]_{IO} \rightarrow L_1$), or deleting (\textit{i.e.}~of the shape $L_1[x:=\epsilon]_{IO}$)

We first introduce a convention on the naming of the symbols on which the IO-substitutions are performed. Given a derivation $d$ for a language $L \in \mathbf{IO(L)}$, one can remark that letters used in the IO-substitution can be renamed under certain constraints. Indeed, given $d_1[x:=d_2]_{IO} \in \mathcal{D}_L$, and $L'_1, \dots, L'_n \in \mathbf{L}$ the languages used in the derivation $d_1$, one can rename $x$ into any letter which has no occurrence in $\bigcup_{1 \leq i \leq n}L'_i$, and ensure that the same language $L$ is derived\footnote{For readers familiar with the $\lambda$-calculus, the precise conditions under which such letters can be renamed are similar to the $\alpha$-equivalence on $\lambda$-terms: variables can be renamed under the constraint that no other variable is ``captured'' by this renaming. We do not detail such constraints in the present work. The analogy with $\lambda$-calculus i made explicit in \cite{Bourreau-IOsubstitution}}. In the rest of the article, we will assume that, without lost of generality, given a language $L \in \mathbf{IO(L)}$ and $d \in \mathcal{D}_L$, each letter on which an $IO$-substitution is performed has a unique occurrence in the IO-substations of $d$, and no occurrence in $L$\footnote{Such a strict convention is to be compared with Barendregt's convention on variables in the $\lambda$-calculus}. This will allow us, in particular, to associate a unique IO-substitution to such a letter.

\begin{definition}[Irrelevant and deleting IO-substitution]
  Given two languages $L_1 \subseteq \Sigma_1^\ast$ and $L_2 \subseteq \Sigma_2^\ast$, we call the IO-substitution $L_1[x:=L_2]_{IO}$:
  \begin{itemize}
    \item an \emph{irrelevant IO-substitution} if for every word $w \in L_1$,  $|w|_x =0$.
    \item a \emph{deleting IO-substitution} if $L_2=\epsilon$.
    \end{itemize}
  \end{definition}

\begin{lemma}\label{lem:no-irrelevant-iosub}
Let us consider a class \textbf{L} of languages, closed under homomorphism, and a language $L \in \mathbf{IO(L)}$. There exists a derivation tree of $L$ with no irrelevant and no deleting substitution.
\end{lemma}
\begin{proof}
  First, if $|w|_x=0$ for every word $w \in L_1$, then $\iosub{L_1}{x}{L_2} \rightarrow L_1$, and such a substitution can be trivially removed from any derivation $d \in \mathcal{D}_L$.
  
  We show that, given a language $L \in \mathbf{IO(L)}$ and $d \in \mathcal{D}_L$, there exists a derivation tree $d'$ in $\mathcal{D}_L$ with no deleting subderivation of the form $L'[z:=\epsilon]_{IO}$. We proceed by induction on $d$: first, if $d=L$, the result is trivial. Otherwise,  $d = d_1[x:=d_2]_{IO}$, where $d_1 \in \mathcal{D}_{L_1}$ and $d_2 \in \mathcal{D}_{L_2}$ for some $L_1, L_2 \in \mathbf{IO(L)}$. If $L_2 \neq \epsilon$ then, by induction hypothesis, there exist $d'_1 \in \mathcal{D}_{L_1}$ and $d'_2 \in \mathcal{D}_{L_2}$ which contain no deleting IO-substitutions; the derivation $d'_1[x:=d'_2]_{IO}$ is still a derivation tree of $L$, and has no deleting IO-substitution. Otherwise, we exhibit such a derivation $d' \in \mathcal{D}_L$ by induction on $d_1$:
    \begin{itemize}
      \item if $d_1=L_1 \in \mathbf{L}$; then, consider an alphabet $\Sigma$ such that $L_1 \subseteq \Sigma^\ast$ and the morphism $io_{x,\epsilon}:\Sigma^\ast \to \Sigma^\ast$. The relation $\iosub{L_1}{x}{\epsilon} \rightarrow L$ is equivalent to $io_{x, \epsilon}(L_1)=L$ by definition; moreover, by the hypothesis that $\mathbf{L}$ is closed by homomorphism, $L$ belongs to \textbf{L}; therefore, there exists $d'=L \in \mathcal{D}_L$.
        \item otherwise $d_1 = d_{11}[y:=d_{12}]_{IO}$ such that $d_{11} \in \mathcal{D}_{L_{11}}, d_{12} \in \mathcal{D}_{L_{12}}, L_{11}, L_{12} \in \mathbf{IO(L)}$ and $\iosub{L_{11}}{y}{L_{12}} \rightarrow L_1$; we have
          \begin{align*}
            \iosub{\iosub{L_{11}}{y}{L_{12}}}{x}{\epsilon} \rightarrow & L & \text{and}\\
            \iosub{\iosub{L_{11}}{x}{\epsilon}}{y}{\iosub{L_{12}}{x}{\epsilon}} \rightarrow & L & \text{by supposing $x \neq y$}
          \end{align*}
          and by induction hypothesis, we know the existence of two derivations $d'_1 \in \mathcal{D}_{L'_{11}}$ and $d'_2 \in \mathcal{D}_{L'_{12}}$ where $\iosub{L_{11}}{x}{\epsilon} \rightarrow L'_{11}$ and $\iosub{L_{12}}{x}{\epsilon} \rightarrow L'_{12}$, such that $d'_{11}$ and $d'_{12}$ contain no deleting IO-substitution. We conclude that $d'_{11}[x:=d'_{12}]_{IO} \in \mathcal{D}_L$ has no deleting IO-substitution.
        \end{itemize}
    \end{proof}

\begin{definition}[Standard derivation]
Given a language $L$ in \textbf{IO(L)}, a derivation $d \in \mathcal{D}_L$ is in \emph{standard form} if:
\begin{itemize}
  \item $d=L$ and $L$ belongs to $\mathbf{L}$, or
  \item $d=\iosub{d_1}{x}{L_2}$ where $L_2$ belongs to $\mathbf{L}$, and $d_1 \in \mathcal{D}_{L_1}$ is in standard form (where $L_1 \in \mathbf{IO(L)}$).
\end{itemize}
\end{definition}

According to this definition, a standard derivation of a language $L \in \mathbf{IO(L)}$ can be written $L_0[x_1:=L_1]_{IO}[x_1:=L_1]_{IO} \dots [x_n:=L_n]_{IO}$, where $n \in \mathbb{N}$ and $L_i \in \mathbf{L}$ for every $1 \leq i \leq n$.
From now on, a standard from will be written according to this notation.

\begin{theorem}\label{th:std-form}
 For every language $L \in \mathbf{IO(L)}$, there exists a standard derivation in $\mathcal{D}_L$.
\end{theorem}
\begin{proof}
Let us consider a derivation $d \in \mathcal{D}_L$. By induction on $d$, we exhibit a standard derivation $d' \in \mathcal{D}_L$. If $d=L$ then $d'=d=L$ is in standard form; otherwise $d=d_1[x:=d_2]_{IO}$, by induction hypothesis, there exists $d'_1=L_{10}[x_1:=L_{11}]_{IO} \dots [x_{n}:=L_{1n}]_{IO} \in \mathcal{D}_{L_1}$ and $d'_2=L_{20}[y_1:=L_{21}]_{IO} \dots [y_{m}:=L_{2m}]_{IO} \in \mathcal{D}_{L_2}$ in standard form. By induction on $m$, we prove the existence of symbols $\{z_1, \dots, z_m\}$ such that:
\begin{center}
  $d'_1[x:=d'_2]_{IO} \in \mathcal{D}_L$ iff $d'_1[x:=L_{20}]_{IO} [y_1:=z_1]_{IO}[z_{1}:=L_{21}]_{IO} \dots [y_m:=z_m]_{IO}[z_{m}:=L_{2m}]_{IO} \in \mathcal{D}_L$
\end{center}
For $0 \leq k \leq m$, let us write $d_{2k}=L_{20}[x_1:=L_{21}]_{IO} \dots [x_{k}:=L_{2k}]_{IO}$.
If $m=0$ the statement is trivial; suppose this is true for $k \in \mathbb{N}$ and $m=k+1$.
Then we have the derivation $d'_1[x:=[d_{2k}[y_{k+1}:=L_{2(k+1)}]_{IO}]_{IO}$ in $\mathcal{D}_L$.
Let us consider a symbol $z_{k+1}$ such that, for every $w \in L_1 \cup \bigcup_{0 \leq i \leq k} L_{2i}$, $|w|_{z_{k+1}}=0$.
Then, given the language $L''$ such that the derivation $\iosub{d_{2k}}{y_{2(k+1)}}{L_{2(k+1)}}$ is in $\mathcal{D}_{L''}$, we know that the derivation  $\iosub{\iosub{d_{2k}}{y_{k+1}}{z_{k+1}}}{z_{2(k+1)}}{L_{2(k+1)}}$ also belongs to $\mathcal{D}_{L''}$;
it follows that $d'_1[x:=d_{2k}]_{IO}[z_{k+1}:=L_{2(k+1)}]_{IO}$ belongs to $\mathcal{D}_L$ because $|w|_x=0$, for all $w \in L_1$. By applying the induction hypothesis on the derivation $d'_1[x:=d_{2k}]_{IO}$, we obtain a derivation in standard form, and the existence of the symbols $\{z_1, \dots, z_m\}$ as stated.
\end{proof}

\begin{theorem}
  Given a language $L  \in \mathbf{IO(L)}$, where $\mathbf{L}$ is a family of languages closed by homomorphism, there is a derivation of $L$ in standard form, and with no irrelevant or deleting IO-substitutions.
\end{theorem}
\begin{proof}
  It suffices to see that the elimination process of irrelevant or deleting substitutions does not modify the structure of the derivation, or equivalently, that the process of creating a derivation in standard form in the proof of theorem~\ref{th:std-form} does not imply the creation of irrelevant or deleting substitutions.
\end{proof}

In the rest of the document, we consider derivations in standard form that contain no irrelevant or deleting IO-subsitutions.


%% file: 4.2-fully-effective.tex
\subsection[The separation lemma]{Fully-effective derivations}

We next characterise a new kind of derivations for languages in $\mathbf{IO(L)}$. This new form will allow us to exhibit only substitution that are effective, \textit{i.e.}~$L_1[x:=L_2]_{IO}$ such that $|w|_x > 0$ for every $w \in L$.
In order to do so, we start by giving a fundamental lemma, which is a direct consequence of the Myhill-Nerode theorem:

\begin{definition}
  Given an alphabet $\Sigma$, a congruence $\cong$ on $\Sigma^\ast$ is an equivalence relation such that, for every $w_1, w_2, u \in \Sigma^\ast$, $w_1 \cong w_2$ implies $w_1u \cong w_2u$.

  Such a congruence is said:
  \begin{itemize}
    \item of finite index if $\Sigma/{\cong}$ is finite.
    \item to saturate a language $L \subseteq \Sigma^\ast$ if for every $w_1, w_2 \in \Sigma^\ast$, $w_1 \cong w_2$ implies $w_1 \in L$ iff $w_2 \in L$ (\textit{i.e.}~$L$ is made of a union of congruence classes in $\Sigma^\ast/{\cong}$)
    \end{itemize}
  \end{definition}

\begin{theorem}[Myhill-Nerode]
  A language $L \subseteq \Sigma^\ast$ is regular iff there exists a congruence $\cong$ of finite index over $\Sigma^\ast$ which saturates $L$.
\end{theorem}

In the following theorem, we are particularly interested in the special case where $L$ is a class of $\Sigma^\ast/{\cong}$, and where $\cong$ is a congruence of finite index on $\Sigma^\ast$.

\begin{corollary}[Separation Lemma\footnote{This lemma is given as the factorisation lemma in \cite{fischerphd}. We use a different name as we already used the terminology of factorisation in section~\ref{part:factorized-parikh}}]\label{lem:separation}
  Consider an alphabet $\Sigma$, a congruence $\cong$ of finite index on $\Sigma^\ast$, a class $C \in \Sigma^\ast/{\cong}$ and a language $L \subseteq \Sigma^\ast$, such that $L$ belongs to a family $\mathbf{L}$ of languages closed by intersection with regular sets. Then $L \cap C$ belongs to $\mathbf{L}$; moreover, $L=\bigcup_{C \in \Sigma^\ast/{\cong}} L \cap C$.
\end{corollary}

The separation lemma is in particular true when the family $\mathbf{L}$ is an abstract family of languages, such as  $\mathbf{RL}, \mathbf{CFL}, \mathbf{yTAL}$ or $\mathbf{MCFL}$.

\begin{lemma}\label{lem:unions-iol}
  Given $L, L_{11}, L_{12}, L_{21}, L_{22} \in \mathbf{IO(L)}$ and $d_{ij} \in \mathcal{D}_{L_{ij}}$ for every $i, j \in \{1, 2\}$:
  \begin{center}$(L_{11} + L_{12}) [x:=L_{21}+L_{22}]_{IO} \in \mathcal{D}_L$ iff $\bigcup_{i, j \in \{1, 2\}} d_{1i}[x:=d_{2j}]_{IO}\in \mathcal{D}_L$\end{center}
\end{lemma}
\begin{proof}
  By definition,we have:
  $$(L_{11} + L_{12}) [x:=L_{21}+L_{22}]_{IO} \in \mathcal{D}_L\text{ iff }\bigcup_{w \in L_{21} + L_{22}} io_{x, w}(L_{11} + L_{12}) = L$$
  Then, we easily establish:
      \begin{eqnarray*}
    \bigcup_{w \in L_{21} + L_{22}} io_{x, w}(L_{11} + L_{12}) & = & \bigcup_{w \in L_{21} + L_{22}} io_{x, w}(L_{11})  \cup \bigcup_{w \in L_{21} + L_{22}} io_{x, w}(L_{12})\\
    & = & \bigcup_{w \in L_{21}} io_{x, w}(L_{11}) \cup \bigcup_{w \in L_{22}} io_{x, w}(L_{11})\\
    & & \cup \bigcup_{w \in L_{21}} io_{x, w}(L_{12}) \cup \bigcup_{w \in L_{22}} io_{x, w}(L_{12})\\
  \end{eqnarray*}
  which is equivalent to $\bigcup_{i, j \in \{1, 2\}} d_{1i}[x:=d_{2j}]_{IO}\in \mathcal{D}_L$.
\end{proof}

\begin{definition}[Fully-effective IO-substitution]
  An IO-substitution $L_1[x:=L_2]_{IO}$ is said \emph{fully-effective} if for every word $w \in L_1$, we have $|w|_x>0$.

  A derivation is said \emph{fully-effective} if every IO-substitution in it is fully-effective.

  Given a language $L \in \mathbf{IO(L)}$, a derivation $d \in \mathcal{D}_L$ is said in \emph{fully effective standard form} when
  $$d=\bigcup_{i \in I} d_{i0}[x_{i_1}:=d_{i1}]_{IO} \dots [x_{i_{n_i}}:=d_{in_i}]_{IO}$$
  where:
  \begin{itemize}
    \item $I$ is a finite set,
    \item for every $i \in I$, the derivation $d_{i0}[x_{i_1}:=d_{i1}]_{IO} \dots [x_{i_{n_i}}:=d_{in_i}]_{IO}$ is fully effective,
    \item and for every $i \in I$, the derivation $d_{i0}[x_{i_1}:=d_{i1}]_{IO} \dots [x_{i_n}:=d_{in_i}]_{IO}$ is in standard form.
 \end{itemize}
\end{definition}

\begin{lemma}\label{lem:fully-effective-std}
Given a family of languages $\mathbf{L}$ closed under intersection with regular sets, for every language $L$ in \textbf{IO(L)}, there exists a derivation in fully effective standard form.
\end{lemma}
\begin{proof}
  Let us consider such a language $L \in \mathbf{IO(L)}$ and a derivation in standard form for it: $L_0[x_1:=L_1]_{IO} \dots [x_n:=L_n]_{IO}$. Given the alphabet $\Sigma$ of $L$, we build the congruence $\cong$ on $(\Sigma \cup \{x_1, \dots, x_n\})^\ast$ such that for every word $w_1, w_2 \in (\Sigma \cup \{x_1, \dots, x_n\})^\ast$1
  $$w_1 \cong w_2 \iff \text{ for every }1 \leq i \leq n, |w_1|_{x_i} =0 \text{ iff } |w_2|_{x_2}=0$$

  This congruence has a finite index $I$ of cardinality $2^n$. For every $1 \leq k \leq n$, according to the separation lemma, we have $L_k=\bigcup_{i\in I} L_{ik}$, where $L_{ik}=L_i \cap C_k$, $C_k$ being the $k^{th}$ class in $(\Sigma \cup \{x_1, \dots, x_n\})^\ast/{\cong}$. Then we can write:
  \begin{align*}
    (\bigcup_{k_0 \in I}L_{0k_0}) [x_1 := L_{1k_1}]_{IO} \dots [x_n:= L_{nk_n}]_{IO} & \rightarrow^\ast L & \\
    \bigcup_{k_0, \dots, k_n \in I} L_{0k_0}[x_1 := L_{1k_1}]_{IO} \dots [x_n:= L_{nk_n}]_{IO} & \rightarrow^\ast L & \hfill\text{(Lemma~\ref{lem:unions-iol})}
  \end{align*}

  According to the separation lemma, $L_{k_0}, \dots, L_{k_n} \in \mathbf{L}$, for every $k_0, \dots, k_n \in I$;
therefore the language $L_{k_0 \dots k_n}$ derived by $L_{0k_0}[x_1 := L_{1k_1}]_{IO} \dots [x_n:= L_{nk_n}]_{IO}$, belongs to \textbf{IO(L)}.
Moreover, for every $0 \leq j \leq n$ and every $k_0, \dots, k_j \in I$, consider the language $L_{k_0 \dots k_j}$ such that $L_{0k_0}[x_1 := L_{1k_1}]_{IO} \dots [x_j:= L_{jk_j}]_{IO} \in \mathcal{D}_{L_{k_0 \dots k_j}}$;
the construction ensures that for every words $w_1$ and $w_2$ in $L_{k_0 \dots k_j}$, $w_1$ and $w_2$ are congruent (it can be proved with a direct induction on $j$);
therefore, the IO-substitution $L_{k_0 \dots k_j}[x_{j+1}:=L_{(j+1)k_{j+1}}]_{IO}$ is fully-effective iff there exists $w$ in $L_{k_0k_1\dots k_j}$ s.t. $|w|_{x_{j+1}}>0$.
According to Lemma~\ref{lem:no-irrelevant-iosub}, we can remove the irrelevant IO-substitutions, which are exactly the susbtitutions which are not fully-effective.
This leads to the existence of a derivation $d_{k_0 \dots k_n}$ in $\mathcal{D}_{L_{k_0 \dots k_n}}$ in fully effective and standard form;
hence, $\bigcup_{k_0, \dots, k_n \in I} d_{k_0 \dots k_n}$ is a derivation in fully-effective standard form for the language $L$.
\end{proof}


%% file: 4.3-a-linearity.tex
\subsection{$a$-linearity}

Finally, we study with more precision the copying effects of the IO-substitution. We already saw how this operation allows one to build non-semilinear languages which verify the constant-growth property. In this section, we study the effect of an IO-substitution on symbols.

Based on the naming convention we adopted for the derivation trees of languages in \textbf{IO(L)}, we first introduce the notion of introducers so as to be able to precisely study the copying process which occurs along a sequence of IO-substitutions.

\begin{definition}[a-introducers]
Let us consider $L \in \mathbf{IO(L)}$, the alphabet $\Sigma$ of $L$, and a derivation $d=L_0[x_1:=L_1]_{IO} \dots [x_n:=L_n]_{IO} \in \mathcal{D}_L$ in standard form. We define the binary relation ${In_d}\subseteq (\Sigma \cup \{x_1, \dots, x_n\}) \times (\Sigma \cup \{x_1, \dots, x_n\})$ by $b {In}_d a$ iff either:
\begin{itemize}
  \item $b=a$ or,
  \item $b=x_i$ for some $1 \leq i \leq n$ and there exists $w \in L_i$ such that $|w|_{a}>0$.
  \end{itemize}
We define ${In}_d^\ast$ as the transitive closure of ${In}_d$, and \emph{the set of introducers of $a \in \Sigma \cup \{x_1, \dots, x_n\}$ in $d$} will be written ${In}_d^a=\{b \in \Sigma \cup \{x_1, \dots, x_n\} \mid b {In}^\ast_d a\}$.
\end{definition}

A symbol $b$ is an introducer of another symbol $a$, iff it is involved in creating occurrences of $a$ at some point during the generation: it is either the symbol $a$ itself, or it introduces a symbol in $\{x_1, \dots, x_n\}$ which will be substituted by a language that contains at least one word in which $a$ has an occurrence. Note that the set ${In}^a_d$ must be finite, as the derivations we consider are finite.

\begin{example}\label{ex:a-introducers}
Let us consider the language represented by the following derivation $d$:
$$x_1x_2x_3[x_1:=a]_{IO}[x_2:=b]_{IO}[x_3:=z^\ast]_{IO}[z:=a^\ast]_{IO}$$
Then we have $x {In}_da$ iff $x \in \{a, z, x_1\}$; moreover, ${In}_d^a = \{a, x_1, z, x_3\}$.
\end{example}

\begin{definition}[Chain of introducers]
Given a standard derivation $d$ of a language $L \in \mathbf{IO(L)}$, we call a finite set $E \subseteq {In}_d^a$ a \emph{chain of introducers of $a$ in $d$}, if there exists $x \in E$  such that, for every $y \in E-\{x, a\}$
\begin{enumerate}
  \item $x {In}_d^\ast y$, and
  \item there exists a unique pair $(y_1, y_2) \in (E-\{y\}) \times (E-\{y\})$ such that $y_1 {In}_t y$ and $y {In}_t y_2$.
\end{enumerate}
Such a chain is said \emph{maximal} iff for every chain $E'$ of introducers of $a$ in $t$, $E \subseteq E'$ implies $E=E'$.
\end{definition}

The set of maximal chains of introducers for a letter $a$ and a standard derivation tree $t$ will be written ${Ch}_d^a$.

\begin{example}
Let us consider the language represented by the derivation in example~\ref{ex:a-introducers}. The chains of $a$-introducers in it are $\{a\}$, $\{a, z\}$, $\{a, z, x_3\}$ and $\{a, x_1\}$. The maximal chains of $a$-introducers are $\{a, z, x_3\}$ and $\{a, x_1\}$.
\end{example}

We now define notions of linearity and universal-linearity for symbols into a language. These definitions are natural extensions of the definitions given in section~\ref{section:parikh}.

\begin{definition}[$a$-linearity]
  Let us consider a language $L \subseteq \Sigma^\ast$ and a letter $a \in \Sigma$. Given $\vect{p}(L)=\bigcup_{1 \leq i \leq k} \mathrm{Im}(F_i)$ the Parikh image of $L$,  we say that $L$ is:
\begin{itemize}
  \item \emph{$a$-constant} if for every $1 \leq i \leq k$, $$F_i(x_1, \dots, x_{n_i})[a]=c_i \in \mathbb{N}$$
  \item \emph{$a$-linear} if for every $1 \leq i \leq k$, $$F_i(x_1, \dots, x_{n_i})[a]=c_i+\sum_{1 \leq j \leq n_i}d_jx_j$$ where $c_i \in \mathbb{N}$ and $d_j \in \mathbb{N}$ for every $1 \leq j \leq n_i$;
  \item \emph{$a$-functional} if for every $1 \leq i \leq k$, $$F_i(x_1, \dots, x_{n_i})[a]=c_i+\sum_{1 \leq j \leq m_i}f_j(x_1, \dots, x_{n_i})$$ where $f_j \in \mathbb{N}^{n_i} \to \mathbb{N}$ for every $1 \leq j \leq m_i$, and $c_i \in \mathbb{N}$.
\end{itemize}
\end{definition}

It is immediate to see that a finite language $L \subseteq \Sigma^\ast$ is $a$-constant for every $a \in \Sigma$; similarly, $L$ is a semilinear language if $L$ is $a$-linear for every $a \in \Sigma$; and universally-linear if $a$-functional for every $a \in \Sigma$, and if the functions $f_j$ in the definition above are $k$-linear for every $1 \leq k \leq n_i$.

Next, we prove some technical lemmas on the conditions under which the property of being constant or linear on a specific symbol is ensured by application of the IO-substitution.

\begin{lemma}\label{lem:conditions-linearity-constant}
  Consider a class of semilinear languages $\mathbf{L}$, languages $L, L_1, L_2 \in \mathbf{IO(L)}$, and a non-irrelevant IO-substitution $L_1[x:=L_2]_{IO}$ such that $L_1[x:=L_2]_{IO} \rightarrow L$. Given $\Sigma$ the alphabet of $L$, a letter $a \in \Sigma$, suppose that $L_1$ and $L_2$ are $a$-linear, and that $L_1$ is $x$-linear:
\begin{enumerate}
  \item if $L_1$ is $x$-constant or $L_2$ is $a$-constant, then $L$ is $a$-linear.
  \item if there exists $w \in L_2$ such that $|w|_a > 0$, then $L$ is $a$-constant iff $L_1$ is $x$-constant and $a$-constant, and $L_2$ is $a$-constant.
  \end{enumerate}
\end{lemma}
\begin{proof}
  Let us consider the Parikh images of $L_1$ and $L_2$, respectively $\bigcup_{1 \leq i \leq n_1}\mathrm{Im}(F_i)$ and $\bigcup_{1 \leq j \leq n_2}\mathrm{Im}(G_j)$.
According to Theorem~\ref{th:IO-factorized}, we know that
$$\vect{p}(L)=\bigcup_{1 \leq i \leq n_1}\bigcup_{1 \leq j \leq n_2}\mathrm{Im}(\mathrm{Wtt}_x(F_i) +\mathrm{On}_x(F_i)G_j)$$

\begin{enumerate}
  \item if $L_1$ is $x$-constant:
    \begin{eqnarray*}
      \mathrm{Im}(\mathrm{Wtt}_x(F_i)[a] +\mathrm{On}_x(F_i)G_j[a]) & = & c_i+\sum_{1 \leq r \leq m_i}d_rx_r + k(c_j+\sum_{1 \leq s \leq m_j}d_sy_s)\\
      & = & (c_i+kc_j) +\sum_{1 \leq r \leq m_i}d_rx_r + \sum_{1 \leq s \leq m_j}d_sy_s
    \end{eqnarray*}
    for every $1 \leq i \leq n_1$ and every $1 \leq j \leq n_2$; it then appears that $L$ is $a$-linear.
    Similarly, if $L_2$ is $a$-constant, we obtain a similar equation, and the same conclusion.
  \item $L$ is $a$-constant iff $\mathrm{Im}(\mathrm{Wtt}_x(F_i)[a] +\mathrm{On}_x(F_i)G_j[a])=c_{ij} \in \mathbb{N}$, for every $1 \leq i \leq n_1$ and every $1 \leq j \leq n_2$, which is verified iff the following equation is true:
    \begin{equation}
      c_{ij} =  c_i+\sum_{1 \leq r \leq m_i}d_rx_r + (c'_i+\sum_{1 \leq r \leq m_i}d'_rx_r)(e_j+\sum_{1 \leq s \leq m_j}k_sy_s)\label{eq:aux}
    \end{equation}
    Under the assumptions that the substitution is not irrelevant, there exists $1 \leq i' \leq n_1$ such that $\mathrm{On}_x(F_{i'})\neq 0$; also, because there exists a word $w \in L_2$ s.t. $|w|_a > 0$, there exists $1 \leq j' \leq n_2$ such that $G_{j'}[a] \neq 0$. Therefore, for every $1 \leq i \leq n_1$ and every $1 \leq j \leq n_2$, $(c'_i+\sum_{1 \leq r \leq m_i}d'_rx_r)(e_j+\sum_{1 \leq s \leq m_j}k_sy_s)$ is constant iff $d'_r=0=k_s$ for every $1 \leq r \leq m_1$ and every $1 \leq j \leq m_2$. Equation~(\ref{eq:aux}) reduces to: $$\mathrm{Im}(\mathrm{Wtt}_x(F_i)[a] +\mathrm{On}_x(F_i)G_j[a]) = c_{i}+c'_ie_j$$ which is equivalent to $L_1$ being both $x$-constant and $a$-constant, and $L_2$ being $a$-constant.
\end{enumerate}
\end{proof}

One should remark that Lemma~\ref{lem:conditions-linearity-constant} cannot be reformulated into an equivalence: indeed, given two semilinear languages $L_1$ and $L_2$, \cite{Bourreau-IOsubstitution} gave the conditions under which a language $L$ s.t. $L_1[x:=L_2]_{IO} \rightarrow L$ is itself semilinear. 

We now give a corollary of this theorem in the particular case of a standard derivation for a language in \textbf{IO(L)}.

\begin{lemma}\label{lem:conditions-linearity}
  Consider a family $\mathbf{L}$ of semilinear languages, a language $L$ in $\mathbf{IO(L)}$ on an alphabet $\Sigma$, a standard derivation $d=L_0[x_1:=L_2]_{IO} \dots [x_n:=L_n]_{IO}$ of $L$ and a letter $a \in \Sigma$.
If, for every chain $C \in Ch_d^a$, there exist at most one $x \in C$ and at most one $0 \leq i \leq n$ such that :
\begin{itemize}
  \item $L_i$ is not $x$-constant, but $y$-constant for every $y \in C-\{x\}$, and
  \item for every $y \in C$ and every $0 \leq j \leq n$ such that $j \neq i$, $L_j$ is $y$-constant
\end{itemize}
then $L$ is $a$-linear.
\end{lemma}
\begin{proof}
Given $d=L_0[x_1:=L_2]_{IO} \dots [x_n:=L_n]_{IO}$, we inductively define $L'_i \in \mathbf{IO(L)}$,  for every $1 \leq i \leq n$ as: $L'_0=L_0$, and $L'_{i-1}[x:=L_i]_{IO} \rightarrow L'_i$. By induction on $i$, we show that $L'_i$ is $x$-linear for every $x \in {In}_d^a$:
\begin{itemize}
  \item if $i =0$, then $L'_i$ is a semilinear language by hypothesis, which implies that $L'_i$ is in particular $x$-linear for every $x \in {In}_d^a$.
  \item suppose the result is true for every $0 \leq k \leq i$. We consider $L'_i[x_{i+1}:=L_{i+1}]_{IO}$;
by induction hypothesis, $L'_i$ is $x$-linear for every $x \in {In}_L^a$, and by hypothesis, $L_{i+1}$ is $c$-linear for every $c \in \Sigma \cup \{x_1, \dots, x_n\}$.
Suppose first that $x_{i+1} \not\in {In}_d^a$; then, for every $c$ such that there exists $w \in L_{i+1}$ for which $|w|_{c} > 0$, we have $c \not\in {In}_L^a$.
Therefore, for every $x \in {In}_d^a$, $L'_{i+1}$ is $x$-linear  iff $L'_i$ is $x$-linear, which is true by hypothesis.
Suppose now that $x_{i+1} \in {In}_d^a$. If $L'_i$ is $x_{i+1}$-constant, according to Lemma~\ref{lem:conditions-linearity-constant}.1, $L'_{i+1}$ is $x$-linear for every $x \in {In}_d^a$;
if $L'_i$ is not $x_{i+1}$-constant, according to Lemma~\ref{lem:conditions-linearity-constant}.1, there exists $y \in {In}_d^{x_{i+1}}$ and $1 \leq j \leq k$ such that $L_j$ is not $y$-constant. Under the constraints in the theorem, we have that for every $x \in {In}_d^a$, $L_{i+1}$ must be $x$-constant. Finally, by application of Lemma~\ref{lem:conditions-linearity-constant}.1, we obtain that $L'_{i+1}$ is $x$-linear for every $x \in {In}_d^a$.
  \end{itemize}
\end{proof}


%% file: 4.4-inverse-homomorphism.tex
\subsection{\textbf{IO(L)} is not an AFL}

Finally, based on the previous results, we prove that,  given $\mathbf{L}$ a separable abstract family of semilinear languages such that $\mathbf{RL} \subseteq \mathbf{L}$, the family $\mathbf{IO(L)}$ is not closed under inverse homomorphism,

We first prove the following lemma, which states that, whenever a chain of IO-substitutions potentially generate a language which is not $a$-linear (\textit{i.e.}~whenever such a chain can  copy words/languages more than once), the derived words must verify a specific pattern.

\begin{lemma}\label{lem:copy-separation}
  Consider a language $L \in \mathbf{IO(L)}$ and:
  \begin{itemize}
    \item a derivation $d=L_0[x_1:=L_1]_{IO} \dots [x_n:=L_n]_{IO} \in \mathcal{D}_L$, fully-effective and in standard form;
    \item a symbol $a \in \Sigma$, a chain $C \in Ch_d^a$, distinct symbols $y_1, y_2 \in C$ and $0 \leq i_1, i_2 \leq n$, $i_1 \neq i_2$ such that, for every $u_1 \in L_{i_1}$ and $u_2 \in L_{i_2}$, $|u_1|_{y_1} > 1$ and $|u_2|_{y_2} > 1$
    \end{itemize}
    Then, for every $w \in L$, there exists $w', w_1, w_2, w_3 \in \Sigma^\ast$ such that $w=w_1aw'aw_2aw'aw_3$.
\end{lemma}
\begin{proof}
First remark, that $y_1 \neq y_2$ and $y_1, y_2 \in C$ implies either $y_1 \in {In}_{d}^{y_2}$ or $y_2 \in {In}_{d}^{y_1}$, by definition of a chain of introducers; let us chose $y_1$ and $y_2$ such that $y_1 \in {In}_{d}^{y_2}$.
We consider a word $u$ in the language derived by $L_0[x_1:=L_1]_{IO} \dots [x_{i_1-1}:=L_{{i_1}-1}]_{IO}$ and a word $u' \in L_{i_1}$.
By hypothesis, $u'$ is of the form $u'_{1}y_1u'_{2}y_1u'_{3}$; because the IO-substitution is not irrelevant, $io_{y_1, u'}(u)$ is of the form $u_1y_1u_2y_1u_3$.

Because there is no deleting IO-substitution, the words in the language derived by $L_{0}[x_{1}:=L_{1}]_{IO} \dots [x_{i_2-1}:=L_{i_2-1}]_{IO}$ must be of the form $w=w_1yw_2yw_3$, where $y_1 {In_d^\ast} y$ and $y {In_d^\ast} y_2$.
Then, because every word in $L_{i_2}$ is of the general form $w'=w'_{1}y_2w'_{2}y_2w'_{3}$, the word $io_{y_2, w'}(w)$ is of the form $w''_1y_2w'_{2}y_2w''_2y_2w'_{2}y_2w''_3$.

Again, because the substitutions are not deleting, we can conclude that the words in $L$ are of the form $u'_1au'au'_2au'au'_3$.
\end{proof}

We now prove our main theorem. The sketch of the proof is similar in many aspects to the proof of the very same non-closure property for IO-macro languages by Fischer.
Indeed, assuming $L_{anp,b} = \{w \in \{a, b\}^\ast \mid |w|_a=nm, \text{ where }n,m > 1\}$ is in $\mathbf{IO(L)}$, and the language $L_{\textit{diff}} = \{b^{p_0}ab^{p_1}a...ab^{p_{nm}} \mid n,m > 1 \text{ and for every }  0 \leq i,j \leq nm, i \neq j \imp p_i \neq p_j\}$, we exhibit a language $L$ such that $L_{\textit{diff}} \subseteq L \subseteq L_{anp,b}$, by removing derivations of Lemma~\ref{lem:copy-separation} in a derivation of $L_{anp,b}$. This means that the IO-substitution is never used in a copying fashion, and therefore, the language derived must be $a$-linear, which is impossible.

\begin{theorem}[Non-closure under inverse homomorphism]
Given an abstract family of semilinear languages $\mathbf{L}$ such that $\mathbf{RL} \subseteq \mathbf{L}$, the family $\mathbf{IO(L)}$ is not closed under inverse homomorphism.
\end{theorem}
\begin{proof}
  Let us consider the language made of a non-prime numbers of $a$
  $$L_{nprime} = \{a^{nm} \mid n,m>1\}$$

  This language is not semilinear since its Parikh image is equal to $\mathrm{Im(F)}$ where $F(x_1, x_2)=4\langle 1 \rangle + x_1\langle 2\rangle + x_2\langle 2\rangle + x_1x_2\langle 1 \rangle$.
  Therefore $L_{nprime}$ does not belong to $\mathbf{L}$, but belongs to $\mathbf{IO(RL)}$: indeed $a^2a^\ast[a:=a^2a^\ast]_{IO} \rightarrow L_{nprime}$ and $a^2a^\ast$ is a regular language. Therefore, if $\mathbf{RL} \subseteq \mathbf{L}$, then $L_{nprime}$ is in $\mathbf{IO(L)}$

  Now, consider the homomorphism $\phi:\{a, b\} \to a^\ast$ such that $\phi(a)=a$ and $\phi(b)=\epsilon$. Then we obtain:
  $$\phi^{-1}(L_{nprime}) = L_{anp,b} = \{w \in \{a, b\}^\ast \mid |w|_a=nm, \text{ where }n,m > 1\}$$
  
  Let us assume $L_{anp, b}$ belongs to $\mathbf{IO(L)}$.
  Then, according to Lemma~\ref{lem:fully-effective-std}, there exists a fully-effective standard derivation $d_{anp,b}=\bigcup_{i \in I}d_i$ for this language,  where for every $i \in I$, $d_i=L_{i0}[x_{i_1}:=L_{i1}]_{IO} \dots [x_{i_{n_i}}:=L_{in_i}]_{IO}$; and for every $0 \leq j \leq n_i$  $L_{ij}$, belongs to $\mathbf{L}$.

  Let us consider the language $L_{\textit{diff}} \subsetneq L_{anp, b}$ defined as:
  $$L_{\textit{diff}} = \{b^{p_0}ab^{p_1}a...ab^{p_{nm}} \mid n,m > 1 \text{ and for every }  0 \leq i,j \leq nm, i \neq j \imp p_i \neq p_j\}$$

We aim at building a language $L$ such that $L_{\textit{diff}} \subseteq L \subseteq L_{anp,b}$. In order to do so, for every $i \in I$, let us consider the congruence $\cong_i$ defined as:
$$w_1 \cong_i w_2 \text{ iff for every }y \in {In_{d_i}^a}, |w_1|>1 \iff |w_2|>1$$

Such a congruence is of finite index. According to the separation lemma and lemma \ref{lem:unions-iol}, we can consider the derivation
$$d'_i = \bigcup_{C_0, \dots, C_{n_i} \in \Sigma^\ast/{\cong_i}} (L_{i0} \cap C_0)[x_1 := (L_{i1} \cap C_1)]_{IO} \dots [x_{n_i} := (L_{in_i} \cap C_{n_i})]_{IO}$$
such that $d_i$ and $d'_i$ derive the same language. Moreover, for every $C_0, \dots, C_{n_i} \in \Sigma^\ast/{\cong_i}$, the derivation $(L_{i0} \cap C_0)[x_1 := (L_{i1} \cap C_1)]_{IO} \dots [x_{n_i} := (L_{in_i} \cap C_{n_i})]_{IO}$ is in standard form, and is fully effective (because $L_{ij} \cap C_j$ is a sublanguage of $L_{ij}$, for every $1 \leq j \leq n_i$).

Now let us consider a derivation $d''_i=(L_{i0} \cap C_0)[x_1 := (L_{i1} \cap C_1)]_{IO} \dots [x_{n_i} := (L_{in_i} \cap C_{n_i})]_{IO}$ for some $C_0, \dots, C_{n_i} \in \Sigma^\ast/{\cong_i}$, such that there exist a chain $ch \in {Ch}_{d''_i}^a$, symbols $y_1, y_2 \in {Ch}_{d''_i}^a$ ($y_1 \neq y_2$), and integers $0 \leq i_1, i_2 \leq n_i$, for which:
\begin{itemize}
  \item for every word $w \in C_{i_1}$, $|w|_{y_1} > 1$;
  \item for every word $w \in C_{i_2}$, $|w|_{y_2} > 1$;
\end{itemize}

Then, according to lemma~\ref{lem:copy-separation}, any word in $L''_i$ where $d''_i \in \mathcal{D}_{L''_i}$ does not belong to $L_{\textit{diff}}$. We can therefore build the language $L$ such that:
$$d = \bigcup_{i \in I'}\bigcup_{C_0 \in \mathcal{C}_{i0}} \dots \bigcup_{C_{n_i} \in \mathcal{C}_{in_i}} (L_{i0} \cap C_0)[x_1 := (L_{i1} \cap C_1)]_{IO} \dots [x_{n_i} := (L_{in_i} \cap C_{n_i})]_{IO}$$
is in $\mathcal{D}_L$, where $d$ results from removing the derivations of languages which intersection with $L_{\textit{diff}}$ is empty.

But, for every $i \in I'$ and every $C_0 \in \mathcal{C}_{i0}, \dots C_{n_i} \in \mathcal{C}_{in_i}$, the derivation $(L_{i0} \cap C_0)[x_1 := (L_{i1} \cap C_1)]_{IO} \dots [x_{n_i} := (L_{in_i} \cap C_{n_i})]_{IO}$ must verify the assumptions of Lemma~\ref{lem:conditions-linearity}; therefore, such a derivation derives a language which is $a$-linear, and $L$ is a finite union of $a$-linear languages, hence an $a$-linear language itself.

But, $\phi(L_{\textit{diff}})=L_{nprime} \subseteq \phi(L) \subseteq \phi(L_{anp, b}) = L_{nprime}$. Therefore, $L_{nprime}$ should be $a$-linear, which is false, and we obtain a contradiction.

\end{proof}

We already commented the analogy between our demonstration and the one in \cite{fischerphd}. One major difference is that Fischer's proof is strongly related to the formalism generating IO-macro languages. In the present case, we intend to work only on the notions of semilinearity and of the copying power which enrich the original family of semilinear languages \textbf{L}.


%% file: 5-conclusion.tex
\section{Conclusion}

In the present paper, we propose a study on the effect of the IO-substitution on the Parikh image of languages in \textbf{L}, an abstract family of semilinear languages.
We first gave a full and complete characterisation of these images in terms of factored Parikh image, and based on this result, we gave a new proof that languages in \textbf{IO(L)} verify the constant-growth property.
This first step was also the opportunity to define universally- and existentially-linear Parikh images, and to prove that languages which Parikh images belong to these classes also verify the constant-growth property.
We gave some brief arguments in favour of the interest of the newly introduced classes of universally- and existentially-linear Parikh images in capturing natural language syntax, which would require further investigations.
In the second part of the paper, we proved that \textbf{IO(L)} is not closed under inverse homomorphism, when $\mathbf{RL} \subseteq \textbf{L}$.
The proof relies on the results obtained in the first section, and in particular in showing that the copying power brought by the IO-substitution operation forces the words to verify a certain pattern.
As a consequence, we can conclude that $\mathbf{IO(MCFL)}$ is not an abstract family of languages, which was an open question in \cite{Bourreau-IOsubstitution}. 

This work gives space for further problems. First, the sketch of the proof of the non-closure property under inverse homomorphism can probably be reused to prove the same result on other formalisms in which copying material is allowed.
In particular, we can conjecture that parallel multiple context-free languages are not closed by such an operation, which contradicts the conjecture in the seminal paper \cite{seki91}. The same question can be addressed on the language in the IO hierarchy \cite{Damm82,salvati-kobele-io}.
Some questions can also be addressed related to the first part of the present article. For instance, how can we generate languages which are full and complete for universally-linear sets?  Addressing the same question on the existentially-linear sets seems less trivial as the functions used to build such sets are free but on one of their arguments.

Finally, some formal questions on the IO-substitution operation can be addressed. One of them is to characterise the languages obtained with infinite application of such an operation; in particular, IO-macro languages might be generated by recursive application of some IO-substitutions. For example, the language $\{a^{n^2} \mid n \in \mathbb{N}\}$ can be expressed as: $\epsilon+a([a:=aa]_{IO})^\ast = \{\epsilon\}\cup \bigcup_{n \in \mathbb{N}}a\underbrace{[a:=aa]_{IO}[a:=aa]_{IO}\dots [a:=aa]_{IO}}_{n}$. With such patterns, one might be able to express languages such as macro-languages, index languages or parallel multiple context-free languages. We will therefore investigate whether the IO-substitution can be used to revisit and classify classes of languages in which some copying mechanism is used.
